\documentclass[11pt] {article}
\usepackage{times}
\usepackage{graphics}
\usepackage{tensor}
\usepackage{natbib}
\usepackage{makeidx}
\usepackage{latexsym}
\usepackage{bm}
\usepackage{amsmath}
\usepackage{amsthm}
\usepackage{amssymb}
\makeindex

\newcommand{\beq}{\begin{equation}}
\newcommand{\eeq}{\end{equation}}
\newcommand{\dif}[2]{\frac{{\rm d} #1}{{\rm d} #2}}

\newcommand{\ildif}[2]{{\rm d} #1/{{\rm d} #2 }}

\newcommand{\ilpdif}[2]{\partial #1/{\partial #2 }}
\newcommand{\pdif}[2]{\frac{\partial #1}{\partial #2}}

\newcommand{\ilpddif}[3]{\partial^2 #1/{\partial #2 \partial #3}}
\newcommand{\defn}{\begin{quote}{\bf Definition. }}
\newcommand{\edefn}{\end{quote}}
\newcommand{\thm}{\begin{theorem}}
\newcommand{\ethm}{\end{theorem}}

\newcommand{\bmat}[1]{\left ( \begin{array}{#1}}
\newcommand{\emat}{\end{array}\right )}
\newcommand{\E}{\mathbb{E}}

\newcommand{\ts}{^{\sf T}} 
\newcommand{\its}{^{\sf -T}}

\newcommand{\bp}{{\bm \beta}}

\newcommand{\vm}{\bm}
\newcommand{\vf}{\bf}

\theoremstyle{definition}

\theoremstyle{plain}
\newtheorem{theorem}{Theorem}

\newcommand{\eps}[3]
{{\begin{center}
 \rotatebox{#1}{\scalebox{#2}{\includegraphics{#3}}}
 \end{center}}
}

\newcommand{\dsp}{1} 

\renewcommand{\baselinestretch}{\dsp}

\begin{document}
\renewcommand{\baselinestretch}{1}
\title{A generalized Fellner-Schall method for smoothing parameter estimation with application to Tweedie location, scale and shape models}
\author{ Simon N. Wood and Matteo Fasiolo\\ School of Mathematics, University of Bristol, Bristol, U.K.\\
{\tt simon.wood@bath.edu}}

\maketitle

\begin{abstract}
We consider the estimation of smoothing parameters and variance components in models with  a regular log likelihood subject to quadratic penalization of the model coefficients, via a generalization of the method of  \cite{fellner1986} and \cite{schall1991}. In particular: (i) we generalize the original method to the case of penalties that are linear in several smoothing parameters, thereby covering the important cases of tensor product and adaptive smoothers; (ii) we show why the method's steps increase the restricted marginal likelihood of the model, that it tends to converge faster than the EM algorithm, or obvious accelerations of this, and investigate its relation to Newton optimization; (iii) we generalize the method to any Fisher regular likelihood. The method represents a considerable simplification over existing methods of estimating smoothing parameters in the context of regular likelihoods, without sacrificing generality: for example, it is only necessary to compute with the same  first and second derivatives of the log-likelihood required for coefficient estimation, and not with the third or fourth order derivatives required by alternative approaches. Examples are provided  which would have been impossible or impractical with pre-existing Fellner-Schall methods, along with an example of a Tweedie location, scale and shape model which would be a challenge for alternative methods. 
\end{abstract}

\renewcommand{\baselinestretch}{\dsp}

\section{Introduction}

This paper is about a very simple method for estimating the smoothing parameters and certain other variance parameters of models with a regular log likelihood, subject to quadratic penalization. The method generalizes the method of \cite{fellner1986} and \cite{schall1991}, by extending the range of smooth model terms with which it can deal, and generalizing beyond the GLM setting to models with any Fisher regular likelihood. The advantage of the Fellner-Schall method is that it offers a simple explicit formula by which smoothing and variance parameters can be iteratively updated using essentially the same quantities anyway required in order to estimate the model coefficients. This has led to its use with smooth additive models, by \cite{rigby2013automatic} amongst others. However the original method has some disadvantages. Firstly it lacks generality, applying only to smooth terms each having a single smoothing parameter, so that tensor product smooth interactions and adaptive smoothers can not be employed. \cite{rodriguez2015sap} partially remove this restriction for some tensor product smooths, but what we propose here is both simpler and more general. Secondly the original method only applies to GLM type likelihoods, with application beyond that setting relying on treating linearized approximations as Gaussian. Again what we propose is simpler and more general. Thirdly the original method derivations, while plausible, do not prove that the method increases the model restricted likelihood at each step, nor offer any insight into convergence rates. We address these issues. In short, it was possible to object that Fellner-Schall methods for updating smoothing parameters were somewhat ad-hoc and insufficiently general. This paper largely removes these objections.

In part, we were motivated to undertake this work by problems in fisheries stock assessment. For example, Figure \ref{mack-space.fig}a shows data from a 2010 survey for mackerel eggs off the coast of western Europe. Such surveys are undertaken in order to help estimate the mass of spawning adults that must be present, and generalized additive models provide suitable spatial models for the mean egg density. As with most fisheries data, the egg counts tend to be highly over-dispersed relative to a Poisson distribution, and a Tweedie distribution \citep{tweedie1984} based model typically offers a much better fit: the variance of a Tweedie random variable $y_i$, with mean $\mu_i$, is given by $\text{var}(y_i) = \phi \mu_i^p$ where $\phi$ and $p$ (here $1<p<2$) are parameters. An important biological feature is that Mackerel are known to favour spawning grounds close to the continental shelf edge, for which the 200m depth contour offers a reasonable proxy. However if mackerel are responding to sea depth, there is no good reason to suppose that this response leads only to a change in the mean density of eggs in the water column: other aspects of the distribution shape are also likely to be effected, and a reasonable model would allow the parameters $p$ and $\phi$ to vary smoothly as sea depth varies. 

\begin{figure}
\vspace*{-.6cm }

\eps{-90}{.6}{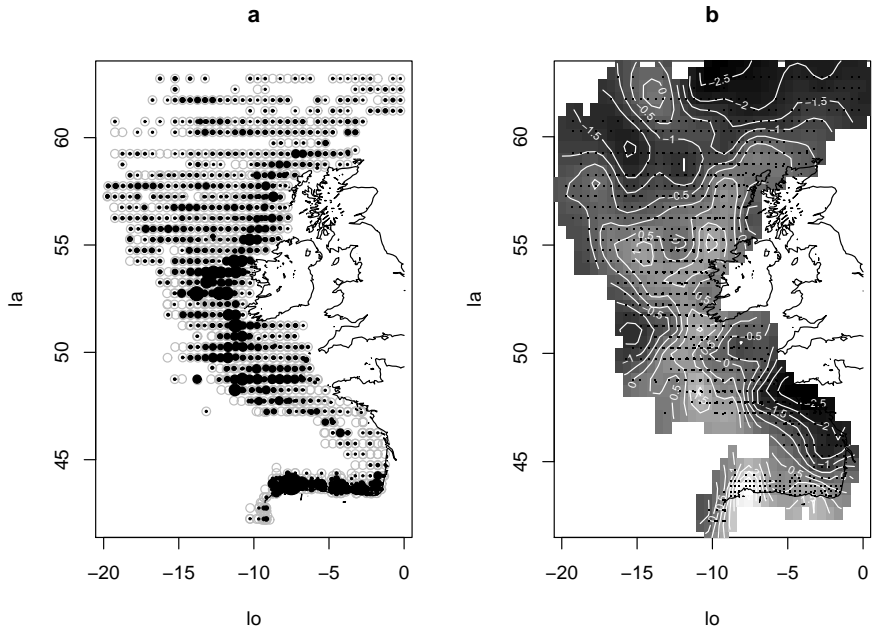}
\vspace*{-.5cm}

\caption{\small {\bf a}. Mackerel ({\em Scomber scombrus}) egg data from the 2010 survey. Grey circles are survey locations, black circles are proportional to the 4th root of egg count. {\bf b}. Image and contour plot of the spatial effect from the Tweedie location scale and shape model described in section \ref{mack.sec}. \label{mack-space.fig}}
\end{figure}

In principle such a model would lie in the GAMLSS class of \cite{rigby2005} and smoothing parameters could be estimated by the method of \cite{wood2015plig}. However, as yet there is no publicly available software for estimating a Tweedie location scale and shape model. The problem is that the Tweedie density does not have an explicit form. Rather, it involves a normalizing constant which is a function of $p$ and $\mu$ and is computable by summing an infinite series `from the middle'. \cite{dunn.smyth2005} provide the details, while \cite{wood2015plig} show how to obtain first and second derivatives of the log density with respect to $p$ and $\mu$: considerable care has to be taken to ensure that the computations maintain numerical stability. The smoothing parameter estimation methods of \cite{wood2015plig} would require third and fourth derivatives of the log Tweedie density, and as yet there are no published methods for stable evaluation of these. Hence it would be useful to have a smoothing parameter estimation method that is general enough to encompass a Tweedie location scale and shape model, while avoiding the need for higher derivatives of the log density.   
 
To introduce the smoothing parameter estimation problem in more detail, first consider the simple case of a Gaussian additive model for a univariate response variable 
\beq
y_i = {\vf A}_i {\vm \theta} + \sum_j g_j(x_{ji}) + \epsilon_i \label{am} 
\eeq
where ${\vf A}_i$ is the $i^{\rm th}$ row of a parametric model matrix, $\vm \theta$ is a vector of unknown coefficients, $g_j$ is a smooth function of (possibly multivariate) covariate $x_j$, and the $\epsilon_i$ are independent $N(0,\sigma^2)$ random deviates. The $g_j$ can be represented using reduced rank spline bases, with associated quadratic penalties penalizing departure from smoothness during fitting. For example $g_j(x) = \sum_k b_k(x) \gamma_k$, where the $b_k$ are spline basis functions and the $\gamma_k$ are coefficients: the associated smoothing penalty is then $\lambda_j {\vm \gamma} \ts {\cal S}_j {\vm \gamma}$, where ${\cal S}_j$ is a fixed matrix, and is usually rank deficient because some functions are treated as `completely smooth'. $\lambda_j$ is a smoothing parameter controlling the strength of penalization during fitting. In general each $g_j$ may have several penalties.  

It is now well established \citep[e.g][]{kimeldorfwahba1970, silverman85, ruppert.wand.carroll} that the smoothing penalties can be viewed as being induced by improper Gaussian prior distributions on the spline coefficients, in which case (\ref{am}) can be re-written as a linear mixed effects model: 
\beq
{\vf y} = {\vf X} {\vm \beta} + {\vm \epsilon}, \text{ } {\vm \beta} \sim N ({\vf 0},{\vf S}_\lambda^-\sigma^2)\text{ and }
{\vm \epsilon} \sim N({\vf 0},{\vf I} \sigma^2), \label{mixed.model}
\eeq
where $\sigma^2$ and $\vm \lambda$ are parameters, $\vm \beta$ is a coefficient vector containing $\vm \theta$ and the coefficients for each smooth term, and $\vf X$ is an $n \times p$ model matrix, containing $\vf A$ and the evaluated basis functions of the smooth terms. ${\vf S}_\lambda$ is a positive semi-definite precision matrix, with Moore-Penrose pseudoinverse ${\vf S}_\lambda^-$. Let ${\vf S}_j$ be ${\cal S}_j$ padded out with zeroes, so that ${\vm \beta} \ts {\vf S}_j {\vm \beta} = {\vm \gamma} \ts {\cal S}_j {\vm \gamma}$, where $\vm \gamma $ is the coefficient vector for $g_j$. Then ${\vf S}_\lambda = \sum_j \lambda_j {\vf S}_j$ (some $g_j$ may each be penalized by several terms in this summation).  The null space of ${\vf S}_\lambda$ is interpretable as the  space of model fixed effects, whereas the range space is the space of random effects. Obviously other simple Gaussian random effect terms can be included in the model in addition to smooth functions. 

\cite{fellner1986} developed a simple iteration for updating $\vm \lambda$ in order to maximize the restricted marginal likelihood of (\ref{mixed.model}), for the special case in which ${\vf S}_\lambda = \sum_j \lambda_j \mathbb{I}_j$, the $\mathbb{I}_j$ being identity matrices with most of their diagonal entries zeroed, and no non-zero entries in common between different $\mathbb{I}_j$. \cite{schall1991} extended this to generalized linear mixed models. Here we first give a simple generalization of the Fellner-Schall method that applies to any model with the structure (\ref{mixed.model}), including smooth additive models in which the smoother terms each have multiple smoothing parameters. We also show why the method improves the restricted marginal likelihood at each step, which is something not revealed by the conventional derivations of the original method. In the additive Gaussian setting our main result is the update formula
$$
\lambda_j^* = \sigma^2\frac{\text{tr}({{\vf S}_\lambda^-{\vf S}_j}) -
\text{tr}\{({\vf X}\ts {\vf X} + {\vf S}_\lambda)^{-1} {\vf S}_j\}}{\hat {\vm \beta} \ts {\vf S}_j \hat {\vm \beta}} \lambda_j.
$$
We also consider updates in the case of any model giving rise to a regular likelihood, but with the previously described prior distribution structure on $\vm \beta$, resulting in the general update (\ref{general-update}) in section \ref{sec.general}: generalized linear mixed models are a special case. In practice the update formula is iteratively alternated with evaluation of $\hat {\vm \beta}$, given the current $\vm \lambda$ estimates. 

The rest of the paper is structured as follows. We first consider the case of Gaussian additive models, deriving a Fellner-Schall type update that can deal with terms with multiple smoothing parameters using a derivation that shows, by construction, that the update must increase the model restricted marginal likelihood of the model. We then study the method in the context of updating one smoothing parameter from a model with several smoothing parameters, showing that it takes longer steps than the EM algorithm, or the most obvious acceleration of the EM algorithm, while not overshooting the maximum of the restricted marginal likelihood, at least in the large sample limit. The update is then generalized to the case of any Fisher-regular likelihood, at the cost of a large sample approximation borrowed from the PQL method. Finally, we present two simple examples which were not possible with previous Fellner-Schall methods, before returning to the Tweedie location scale and shape model for the Mackerel data.

\section{Why the modified update works \label{sec-gaussian}}

For model (\ref{mixed.model}), the improper log joint density of the data, $\vf y$, and coefficients, $\vm \beta$, can be written as 
$$
\log f_\lambda ({\vf y},{\vm \beta}) = - \frac{\|{\vf y}-{\vf X} {\vm \beta}\|^2 + {\vm \beta} \ts {\vf S}_\lambda {\vm \beta}}{2 \sigma^2} + \log |{\vf S}_\lambda/\sigma^2|_+  /2 + c  
$$
where $|{\vf S}_\lambda|_+$ denotes the product of the non-zero eigenvalues of ${\vf S}_\lambda$ and we use $c$ to denote a parameter independent constant, which may vary from expression to expression. Following \cite{wood2011} the log restricted marginal likelihood can conveniently be written as 
$$
l_r({\vm \lambda}) = -\frac{\|{\vf y}-{\vf X} \hat {\vm \beta}_\lambda\|^2 + \hat {\vm \beta}_\lambda \ts {\vf S}_\lambda \hat {\vm \beta}_\lambda}{2 \sigma^2} + \log |{\vf S}_\lambda/\sigma^2|_+/2 - \log |{\vf X}\ts {\vf X}/\sigma^2 + {\vf  S}_\lambda/\sigma^2 |/2  + c
$$
where $\hat {\bm \beta}_\lambda = \text{argmax} f_\lambda ({\vf y},{\vm \beta})$ for a given $\vm \lambda$. Expressing the joint density and $l_r$ in this way is the key to straightforwardly obtaining a general update formula. Given that $\ilpdif{(\|{\vf y}-{\vf X} {\vm \beta}\|^2 + {\vm \beta} \ts {\vf S}_\lambda {\vm \beta})}{\vm \beta}|_{\hat \beta_\lambda} = {\vf 0}$, by definition of $\hat {\vm \beta}_\lambda$, we have
$$
\pdif{l_r}{\lambda_j} = \text{tr}({{\vf S}_\lambda^-{\vf S}_j})/2 -
\text{tr}\{({\vf X}\ts {\vf X} + {\vf S}_\lambda)^{-1} {\vf S}_j\}/2 -
\hat {\bm \beta}_{\lambda} \ts{\vf S}_j \hat {\bm \beta}_{\lambda} /(2\sigma^2).
$$
If we were to follow the conventional derivation of the Fellner-Schall method, we would now multiply all terms in $\ilpdif{l_r}{\lambda_j}$ by $\lambda_j$, and then decide to treat two of these $\lambda_j$ as fixed at their previous estimate, while one is to be updated. Equating $\ilpdif{l_r}{\lambda_j}$ to zero and re-arranging then gives the update equation. Such an approach does not reveal why the update increases $l_r$, so we instead give an alternative derivation.

By Theorem 1, below, $\text{tr}({{\vf S}_\lambda^-{\vf S}_j}) -
\text{tr}\{({\vf X}\ts {\vf X} + {\vf S}_\lambda)^{-1} {\vf S}_j\}$ is non-negative, while $\hat {\bm \beta}_{\lambda} \ts {\vf S}_j \hat {\vm \beta}_{\lambda}$ is non-negative by the positive semi-definiteness of ${\vf S}_j$. Hence $\ilpdif{l_r}{\lambda_j}$ will be negative if 
$$
\text{tr}({{\vf S}_\lambda^-{\vf S}_j}) -
\text{tr}\{({\vf X}\ts {\vf X} + {\vf S}_\lambda)^{-1} {\vf S}_j\} < \hat {\vm \beta}_\lambda \ts {\vf S}_j \hat {\vm \beta}_{\lambda}/\sigma^2,
$$
indicating that $\lambda_j$ should be decreased. If the inequality is reversed then $\ilpdif{l_r}{\lambda_j}$ is positive, indicating that $\lambda_j$ should be increased. If the inequality becomes an equality then $\ilpdif{l_r}{\lambda_j}=0$ and $\lambda_j$ should not be changed. A final requirement of any update is that $\lambda_j$ should remain positive. A simple update that clearly meets all four requirements is  
\beq
\lambda_j^* = \sigma^2\frac{\text{tr}({{\vf S}_\lambda^-{\vf S}_j}) -
\text{tr}\{({\vf X}\ts {\vf X} + {\vf S}_\lambda)^{-1} {\vf S}_j\}}{\hat {\vm \beta}_\lambda \ts {\vf S}_j \hat {\vm \beta}_\lambda} \lambda_j, \label{g-schall.update}
\eeq
with $\lambda_j^*$ set to some pre-defined upper limit if $\hat {\vm \beta}_\lambda \ts {\vf S}_j \hat {\vm \beta}_\lambda$ is so close to zero that the limit would otherwise be exceeded. Formally ${\vm \Delta} = {\vm \lambda}^* - {\vm \lambda}$ is an ascent direction for $l_r$, by Taylor's theorem and the fact that 
${\vf \Delta} \ts \ilpdif{l_r}{{\vm \lambda}} > 0$,  unless $\vm \lambda$ is already a turning point of $l_r$. To formally guarantee that the update increases $l_r$ requires step length control, for example we use the update ${\vm \delta} = {\vm \Delta}/2^k$, where $k$ is the smallest integer $\ge 0$ such that $l_r({\vm \lambda} + {\vm \delta}) > l_r({\vm \lambda})$.

Two terms in the update have the potential to be of $O(p^3)$ floating point cost, but $\text{tr}\{({\vf X}\ts {\vf X} + {\vf S}_\lambda)^{-1} {\vf S}_j\}$ can re-use the Cholesky factor of ${\vf X}\ts {\vf X} + {\vf S}_\lambda$, which is anyway required to estimate $\hat {\bm \beta}_\lambda$, while the block diagonal nature of ${\vf S}_\lambda$ means that in reality 
$\text{tr}({{\vf S}_\lambda^-{\vf S}_j})$ has $O(q_j^3)$ computational cost, where $q_j$ is the number of coefficients affected by ${\vf S}_j$ and is typically far fewer than $p$. Under the conditions of the original Fellner-Schall proposal, then $\text{tr}({{\vf S}_\lambda^-{\vf S}_j}) = \text{rank}({\vf S}_j)/\lambda$ and we recover exactly the Fellner-Schall update, albeit with a slightly more computationally tractable expression. The update relies on the following, which is the key to the generalization beyond singly penalized smooth terms.

\begin{theorem} Let $\vf B$ be a positive definite matrix and ${\vf S}_\lambda$ be a positive semi-definite matrix parameterized by $\vm \lambda$, and with a null space that is independent of the value of $\vm \lambda$. Let positive semi-definite matrix ${\vf S}_j$ denote the derivative of ${\vf S}_\lambda$ with respect to $\lambda_j$. Then
 $\text{tr}({{\vf S}_\lambda^-{\vf S}_j}) -
\text{tr}\{({\vf B} + {\vf S}_\lambda)^{-1} {\vf S}_j\} > 0$.
\end{theorem}
\begin{proof}
Let ${\vf B}={\vf U}{\vm \Lambda} {\vf U}\ts$ be the eigen-decomposition of $\vf B$. If ${\vf S}^\prime_\lambda = {\vm \Lambda}^{-1/2} {\vf U} \ts {\vf S}_\lambda {\vf U} {\vm \Lambda}^{-1/2}$ while ${\vf S}^\prime_j = {\vm \Lambda}^{-1/2} {\vf U} \ts {\vf S}_j {\vf U} {\vm \Lambda}^{-1/2}$ then it follows that $\text{tr}\{({\vf B} + {\vf S}_\lambda)^{-1} {\vf S}_j\} = \text{tr}\{({\vf I} + {\vf S}_\lambda^\prime)^{-1} {\vf S}_j^\prime\}$, while 
$ \text{tr}({{\vf S}_\lambda^-{\vf S}_j}) = \text{tr}({{\vf S}_\lambda^{ \prime -} {\vf S}_j^\prime})$, where ${\vf S}_\lambda^{ \prime -} = {\vm \Lambda}^{1/2} {\vf U} \ts {\vf S}_\lambda^- {\vf U} {\vm \Lambda}^{1/2}$. Now form the second  eigen-decomposition ${\vf S}_\lambda^\prime = {\bf VDV}\ts$. We have that 
$ \text{tr}\{({\vf I} + {\vf S}_\lambda^\prime)^{-1} {\vf S}_j^\prime\} = \text{tr}\{({\vf I}+{\vf D})^{-1} {\vf V}\ts {\vf S}^\prime_j {\vf V} \} $, while $\text{tr}({{\bf S}_\lambda^{ \prime -} {\vf S}_j^\prime}) = \text{tr}({\vf D}^- {\vf V}\ts {\vf S}^\prime_j {\vf V})$. Let $s_i$ denote the diagonal elements of ${\vf V}\ts {\vf S}^\prime_j {\vf V}$. By the conditions of the theorem the null space of ${\vf S}_\lambda$ is independent of $\vm \lambda$, and hence $s_i=0$ if $D_{ii} = 0$. So if $M = \{i; s_i \ne 0\}$, 
$ \text{tr}({{\vf S}_\lambda^-{\vf S}_j}) = \sum_{i \in M} s_i/D_{ii}$ while $\text{tr}\{({\vf B} + {\vf S}_\lambda)^{-1} {\vf S}_j\} = \sum_{i \in M} s_i/(D_{ii}+1)$. Since all the $D_{ii}$ in the summations are positive, by the positive semi-definiteness of ${\vf S}_\lambda$ and the definition of $M$, then the terms in the second summation are each smaller than the corresponding term in the first, and the result is proved.    
\end{proof}

The variance parameter $\sigma^2$ also has to be estimated, but by setting the derivative of $l_r$ with respect to $\sigma^2$ to zero and solving we obtain 
$$
\hat \sigma^2 = \|{\vf y}-{\vf X}\hat {\vm \beta}_\lambda \|^2/[n - \text{tr}\{({\vf X}\ts {\vf X} + {\vf S}_\lambda)^{-1}{\vf X}\ts {\vf X}\}].
$$

\subsection{Comparison with the EM algorithm and Newton optimization \label{em.sec}}

\begin{figure}

\eps{-90}{.7}{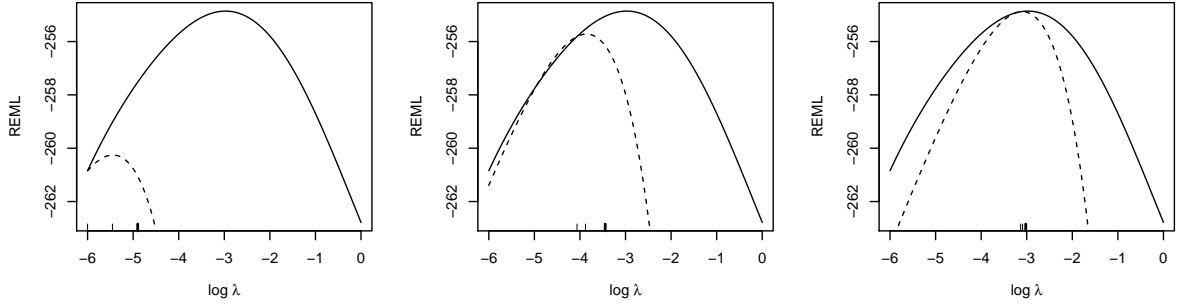}
\vspace*{-.5cm}

\caption{\small Alternate steps of update (\ref{g-schall.update}) for a rank 20 cubic spline smoother of Gaussian data. Each panel shows the log restricted likelihood as a continuous curve, while the EM Q-function is plotted as a dashed curve, shifted to match the log restricted likelihood at each step's start. The two thin ticks on the x axis show the start of the step and the maximum of the Q function. The thick black tick is update (\ref{g-schall.update}). \label{em.fig}}
\end{figure}

The update (\ref{g-schall.update}) can be viewed as a crude approximation to an EM update \citep{dempster1977EM}. Specifically, the EM Q-function for model (\ref{mixed.model}) has the form 
\beq
Q_{\lambda^\prime}({\vm \lambda}) = -\frac{ \|{\vf y}-{\vf X} \hat {\vm \beta}_{\lambda^\prime}\|^2 + \hat {\vm \beta}_{\lambda^\prime} \ts {\vf S}_\lambda \hat {\vm \beta}_{\lambda^\prime}}{2 \sigma^2} + \log |{\vf S}_\lambda/\sigma^2|_+/2 - \text{tr}\{({\vf X}\ts {\vf X} + {\vf S}_{\lambda^\prime})^{-1}{\vf S}_\lambda\}/2, \label{qfunc}
\eeq
and (\ref{g-schall.update}) would be the exact maximiser of $Q$, if $\text{tr}({{\vf S}_\lambda^-{\vf S}_j}) - \text{tr}\{({\vf X}\ts {\vf X} + {\vf S}_{\lambda^\prime})^{-1} {\vf S}_j\} \propto 1/\lambda_j$. 

In fact update (\ref{g-schall.update}) systematically makes larger changes to $\vm \lambda$ than the EM update, as illustrated in Figure \ref{em.fig}. For insight into why this happens consider updating a single $\lambda_j$ relating to a block $\lambda_j {\vf S}_j$ of ${\vf S}_\lambda$, so that $\text{tr}({{\vf S}_\lambda^-{\vf S}_j}) = k/\lambda_j$, where $k = \text{rank}({\vf S}_j)$. Then defining $\gamma = \text{tr}\{({\vf X}\ts {\vf X} + {\vf S}_{\lambda^\prime})^{-1} {\vf S}_j\}$ and $b = \hat {\vm \beta}_{\lambda^\prime} \ts {\vf S}_j \hat {\bm \beta}_{\lambda^\prime}/\sigma^2$, (\ref{g-schall.update}) seeks $\lambda_j$ to solve $k/\lambda_j = b + \gamma \lambda_j^\prime/\lambda_j$, whereas an EM step seeks $\lambda_j$ to solve $k/\lambda_j = b + \gamma $. If  $k/\lambda_j > b + \gamma$ then $\lambda_j$ has to be increased from $\lambda_j^\prime$ under either update. It has to be increased by more under (\ref{g-schall.update}), because $\gamma \lambda_j^\prime/\lambda_j$ decreases monotonically from $\gamma$ as $\lambda_j$ increases from $\lambda_j^\prime$. A similar argument shows that if $k/\lambda_j < b + \gamma$ then the required reduction in $\lambda_j$ is larger under (\ref{g-schall.update}) than under EM. Figure \ref{qsol.fig} shows the root finding problem corresponding to the EM update as a dashed curve, and corresponding to update (\ref{g-schall.update}) as a solid curve, for the same problem illustrated in Figure \ref{em.fig}. 

Figure \ref{em.fig} also illustrates the equivalent problem for the restricted marginal likelihood itself, which can be viewed as solving the same problem as the EM update, but with both $b$ and $\gamma$ being functions of $\lambda$: the dependence of $b$ on $\lambda$ is indirect via $\hat \beta_\lambda$, but the dependence of $\gamma$ is direct. This suggests using an accelerated EM update seeking to solve
$$
k/\lambda_j = b + \gamma(\lambda_j)
$$ 
where $\gamma(\lambda_j) = \text{tr}\{({\vf X}\ts {\vf X} + {\vf S}_{\lambda})^{-1} {\vf S}_j\}$. This obviously makes longer steps than the original EM update, as is illustrated by the dashed curve in Figure \ref{qsol.fig}. Update (\ref{g-schall.update}) also results in longer update steps than this accelerated EM step, as Figure \ref{qsol.fig} suggests and the following demonstrates.

\begin{figure}

\eps{-90}{.5}{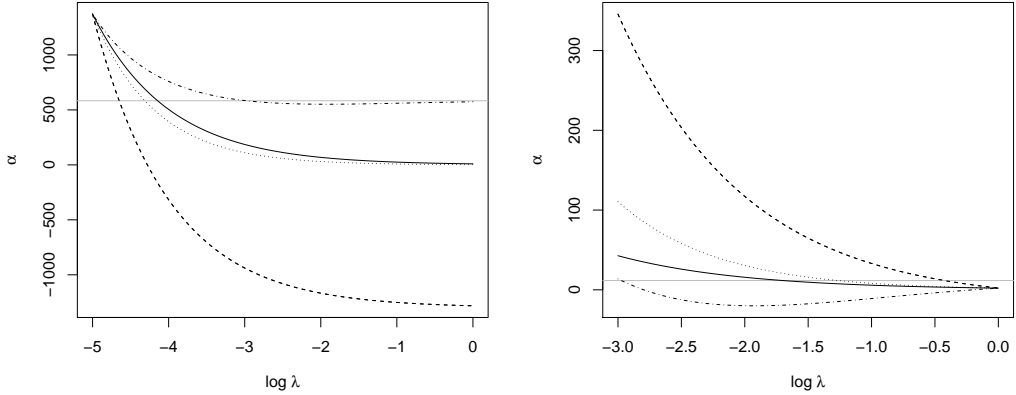}
\vspace*{-.5cm}

\caption{\small Illustration of the root finding problem corresponding to the various updates discussed in section \ref{em.sec}, for the same problem illustrated in Figure \ref{em.fig}. The grey horizontal line is the constant $b$. The right plot corresponds to $\log \lambda_j^\prime = -5$ and the right to  $\log \lambda_j^\prime = 0$. The EM update corresponds to the point at which the dashed curve crosses the $b$ line. The accelerated EM update corresponds to where the dotted curve crosses the b line. Update (\ref{g-schall.update}) corresponds to where the solid curve crosses the $b$ line. The REML optimum is where the dot-dashed curve crosses the $b$ line.
 \label{qsol.fig}}
\end{figure}

\begin{theorem} Consider updating a single $\lambda_j$ corresponding to a diagonal block $\lambda_j {\vf S}_j$ of ${\vf S}_\lambda$. Update (\ref{g-schall.update}) takes a longer step than the equivalent accelerated EM update.  
\end{theorem}
\begin{proof} 
Under the stated conditions $\text{tr}({{\vf S}_\lambda^-{\vf S}_j}) = k/\lambda_j$ where $k=\text{rank}({\vf S}_j)$. Let $\gamma(\lambda_j) = \text{tr}\{({\vf X}\ts {\vf X} + {\vf S}_\lambda)^{-1} {\vf S}_j\}$ and $\alpha(\lambda_j) = k/\lambda_j- \gamma(\lambda_j)$. The accelerated EM step seeks $\lambda_j$ such that $\alpha(\lambda_j) = b $ where $b=\hat {\vm \beta}_{\lambda^\prime} \ts {\vf S}_j \hat {\vm \beta}_{\lambda^\prime}/\sigma^2$, increasing $\lambda_j$ if $\alpha(\lambda_j) > b $ and decreasing $\lambda_j$ if $\alpha(\lambda_j) < b $. Update (\ref{g-schall.update}) is exactly equivalent to seeking $\lambda_j $ such that $\alpha^\prime(\lambda_j) = b$, where $\alpha^\prime(\lambda_j) = k/\lambda_j - \gamma(\lambda_j^\prime)\lambda^\prime_j/\lambda_j$. By definition $\alpha^\prime(\lambda_j^\prime) = \alpha(\lambda_j^\prime)$, so to prove the result it suffices to prove that $\alpha^\prime (\lambda_j) > \alpha(\lambda_j)$ when $\lambda_j>\lambda_j^\prime$ and $\alpha^\prime (\lambda_j) < \alpha(\lambda_j)$ when $\lambda_j<\lambda_j^\prime$.  Now let ${\vf S}_{-j} = \sum_{i\ne j} \lambda_i {\vf S}_i$, and let $\vf B$ be any matrix such that ${\vf B}\ts {\vf B} = {\vf S}_{-j}$. Consider the QR decomposition 
$
({\vf X}\ts, {\vf B}\ts ) = {\vf R}\ts {\vf Q}\ts
$
and form the symmetric semi-definite eigen-decomposition ${\vf U}{\vm \Lambda} {\vf U} \ts = {\vf R}\its  {\vf S}_j {\vf R}^{-1}$.
Routine manipulation shows that $\gamma(\lambda_j) = \sum_{i=1}^k \Lambda_i/(1+\lambda_j \Lambda_i)$. It follows that $\gamma^\prime (\lambda_j) = \sum_{i=1}^k \Lambda_i/(\lambda_j/\lambda_j^\prime + \lambda_j \Lambda_i)$. Hence $\gamma^\prime(\lambda_j) < \gamma(\lambda_j)$ if $\lambda_j > \lambda_j^\prime$ and  $\gamma^\prime(\lambda_j) > \gamma(\lambda_j)$ if $\lambda_j < \lambda_j^\prime$, proving the result. 
\end{proof}

Taking longer steps than a plain or accelerated EM algorithm would be of limited utility if those steps overshot the maximum of the restricted likelihood and require repeated step length control, especially when close to the optimum. In practice such overshoot does not occur. The following theorem offers some insight into the reasons, albeit only asymptotically. 

We assume infil asymptotics and require two technical assumptions. 

\noindent {\bf Assumption 1}: If ${\vf Q}_1$ denotes the first $n$ rows of $\vf Q$ from theorem 2 and ${\vf a} = {\vf U}\ts {\vf Q}_1 \ts {\vf y}$, then $a_i^2 = O_p(n^{\beta_i})$ where $\beta_i > 0$ for all $i$. 

\noindent The assumption is less obscure than it at first appears. To see this, first consider the very mild assumption that the model is sufficiently reasonable that ${\vf y}\ts \hat {\vm \mu}_0 = O_p(n)$, where $\hat {\vm \mu}_0 = {\vf X} \hat {\vm \bp}$, when $\lambda_j=0$. In fact $\hat {\vm \mu}_0 = {\vf Q}_1 {\vf U U}\ts {\vf Q}_1\ts {\vf y}$, and so ${\vf a}\ts {\vf a} = {\vf y}\ts \hat {\vm \mu}_0 $ and  $\text{mean}(a_i^2) = O_p(n/p)$ where $p$ is the fixed dimension of $\vf a$. Now let ${\vf C}= {\vf Q}_1 {\vf U}$, so that $\hat {\vm \mu}_0 = {\vf CC}\ts {\vf y} = \sum_{i=1}^p {\vf C}_{\cdot i}{\vf C}_{\cdot i}\ts {\vf y}$. $\hat {\vm \mu}_0$ can be decomposed into $p$ components $\hat {\vm \mu}_0 = \sum_i \hat {\vm \mu}_i$, where $\hat {\vm \mu}_i = {\vf C}_{\cdot i}{\vf C}_{\cdot i}\ts {\vf y}$. The assumption that ${\vf y} \ts \hat {\vm \mu}_i = O_p(n)$ is essentially equivalent to assuming that no model component is orthogonal to $E({\vf y})$, but since $a_i = {\vf C}_{\cdot i}\ts {\vf y}$, this assumption is equivalent to $a_i^2 = O_p(n)$. So Assumption 1 is reasonable, and in most cases we expect $\beta_i=1$.

\noindent {\bf Assumption 2}: In the notation of theorem 2, $\lambda \Lambda_i = O_p(n^{\alpha_i})$, where $\alpha_i$ is an unknown real constant. 

\noindent This simply assumes that each $\lambda \Lambda_i$ has some polynomial dependence on $n$, but not that we know what it is. 

\begin{theorem} Let the setup be as in theorem 2, and let $\hat \lambda_j$ denote the maximizer of the restricted likelihood with respect to $\lambda_j$. Given assumptions 1 and 2, and for an initial  $\lambda_j$ sufficiently close to $\hat \lambda_j$, the update, $\lambda_j^*$, given by (\ref{g-schall.update}) is either between $\lambda_j $ and $\hat \lambda_j$, or tends to $\hat \lambda_j$ as $n \to \infty$.
\end{theorem}
\begin{proof} Dropping the subscript $j$, let $\rho = \log \lambda$, and let $\lambda$ denote the $j^{\rm th}$ smoothing parameter at the start of the updates. Consider again the root finding problems equivalent to the update (\ref{g-schall.update}) and to maximization of the restricted marginal likelihood. Applying Taylor's theorem to the components of these root finding problems, we have that for $\lambda$ sufficiently close to $\hat \lambda$, 
$$
\frac{k}{\lambda} - \frac{k}{\lambda} (\hat \rho - \rho) - \gamma(\lambda) - \left (\dif{\gamma}{\rho}  + \dif{b}{\rho}\right ) (\hat\rho - \rho ) = b(\lambda)
$$ 
where the derivatives are evaluated at the initial value, $\rho$, and
$$
\frac{k}{\lambda} - \frac{k}{\lambda} (\rho^* - \rho) - \gamma(\lambda)  + \gamma(\lambda)(\rho^* - \rho ) = b(\lambda).
$$ 
So, if $\gamma(\lambda) \le \delta(\lambda) = -(\ildif{\gamma}{\rho}+\ildif{b}{\rho})$, then $\lambda < \lambda^* \le \hat \lambda$. Also, if $\lambda \gamma(\lambda) \to 0$ and $\lambda \delta(\lambda) \to 0$, as $n \to \infty$, then $|\rho^* - \hat \rho| \to 0$.

Now consider the actual behaviour of $\gamma(\lambda)$ and $\delta(\lambda)$.
Using the QR and eigen-decomposition steps of theorem 2, some routine manipulation yields $$
\gamma (\lambda) = \frac{1}{\lambda} \sum_i \frac{\lambda \Lambda_i}{1 + \lambda \Lambda_i} \text{   and   } \delta(\lambda) = \frac{1}{\lambda} \sum_i \frac{\lambda \Lambda_i}{1 + \lambda \Lambda_i} \left \{
\frac{(1 + 2 a_i^2) \lambda \Lambda_i +  \lambda^2 \Lambda_i^2}{1 + 2 \lambda \Lambda_i + \lambda^2 \Lambda_i^2}
\right \}
$$ 
So the $i^{\rm th}$ term of $\delta$ will be larger that the $i^{\rm th}$ term of $\gamma$ if $\lambda \Lambda_i > (2 a_i^2 - 1)^{-1}$: if $\lambda \Lambda_i = O_p(n^{\alpha_i})$ this dominance occurs in the $n \to \infty$ limit when $\alpha_i > -\beta_i$. Furthermore if $\alpha_i < - \beta_i/2$ then the $i^{\rm th}$ terms of $\gamma\lambda$ and $\delta\lambda$ both tend to zero in the large sample limit. So in the large sample limit, sufficiently close to $\hat \lambda$,  there are only two non-exclusive possibilities:  $\gamma(\lambda) < \delta(\lambda) $ so that $\lambda^*$ lies between $\lambda$ and $\hat \lambda$, and/or $\lambda \delta, \lambda \gamma \to 0$ so that $|\lambda^*-\hat \lambda| \to 0$.  
\end{proof}

The solution of the linearised root finding problem corresponding to the restricted likelihood maximisation is the Newton method update. So a corollary of theorem 3 is that iteration of update (\ref{g-schall.update}) will generally converge more slowly than Newton's method, when close to the optimum, and certainly no faster.

\section{Beyond the linear Gaussian case \label{sec.general}}

Now consider replacing the Gaussian log likelihood with another log likelihood, $l$, meeting the Fisher regularity conditions, so that the improper log joint density becomes
$$
\log f_\lambda ({\vf y},{\vm \beta}) = l({\vm \beta}) - {\vm \beta} \ts {\vf S}_\lambda {\vm \beta}/2 + \log |{\vf S}_\lambda|_+  + c 
$$
and in the large sample limit ${\vm \beta}|{\vf y} \sim N(\hat {\vm \beta}_\lambda,{\vf V}_\lambda)$ where ${\vf V}_\lambda^{-1} = {\cal H}_\lambda$ or $\E {\cal H}_\lambda$ and $ {\cal H}_\lambda = -\ilpddif{l}{{\vm \beta}}{{\vm \beta}\ts} + {\vf S}_\lambda$. Newton's method can be used to find $\hat {\vm \beta}_\lambda$, with the usual modifications to guarantee convergence \citep[e.g.][\S 5.1.1]{wood2015core}. Following \cite{wood2015plig} the log Laplace approximate marginal likelihood in this case is conveniently expressed as 
$$
l_r = l(\hat {\vm \beta}_\lambda) - \hat {\vm \beta}_\lambda \ts {\vf S}_\lambda \hat {\vm \beta}_\lambda/2 +  \log{|{\vf S}_\lambda|_+}/2 -  \log{|{\cal H}_\lambda|}/2 + c.
$$   
Defining ${\vf H} = -\ilpddif{l}{{\vm \beta}}{{\vm \beta}\ts}$, we have
$$
\pdif{l_r}{\lambda_j} = - \hat {\vm \beta}_{\lambda} \ts {\vf S}_j \hat {\vm \beta}_{\lambda} /2 + \text{tr} ({{\vf S}_\lambda^-{\vf S}_j})/2 - \text{tr}\{{\vf V}_\lambda {\vf S}_j\}/2 - \text{tr}\{{\vf V}_\lambda \ilpdif{{\vf H}}{\lambda_j}\}/2.
$$
The direct dependence of  $\ilpddif{l}{{\vm \beta}}{{\vm \beta}\ts}$ on $\lambda_j$ is inconvenient. However the PQL and performance oriented iteration methods for $\vm \lambda$ estimation of \cite{breslow.clayton} and \cite{gu.perf.iter} both neglect the dependence of $\ilpddif{l}{{\vm \beta}}{{\vm \beta}\ts}$ on $\vm \lambda$, on the basis that it anyway tends to zero in the large sampe limit. If we follow these precedents then the development follows the Gaussian case and the update is 
\beq
\lambda_j^* = \frac{\text{tr}({{\vf S}_\lambda^-{\vf S}_j}) -
\text{tr}\{{\vf V}_{\lambda^\prime} {\vf S}_j\}}{\hat {\vm \beta}_\lambda \ts {\vf S}_j \hat {\vm \beta}_\lambda} \lambda_j. \label{general-update}
\eeq
If $\ilpddif{l}{{\vm \beta}}{{\vm \beta}\ts}$ is independent of $\vm \lambda$ at finite sample size, as is the case for some distribution -- link function combinations in a generalized linear model setting, then the update is guaranteed to increase $l_r$ under step size control, but otherwise this is not the case, and in practice the $\vm \lambda $ estimate no longer exactly maximizes $l_r$. 

Theorem 1, required to guarantee that $\lambda_j^*>0$, will hold if $V_\lambda$ is based on the expected Hessian of the negative log likelihood, but if it is based on the observed Hessian, then this must be positive definite for the Theorem to hold. Hence, if the observed Hessian is not positive definite then the expected Hessian, or a suitable nearest positive definite matrix to the observed Hessian, should be substituted.

As in the Gaussian case a link to the EM update can again be established via an approximate $Q$ function, obtained by taking a second order Taylor expansion of $l$ around $\hat {\vm \beta}_\lambda$, and using the large sample distribution of ${\vm \beta} | {\vf y}$: 
$$
Q_{\lambda^\prime}^*({\vm \lambda}) = l(\hat {\vm \beta}_{\lambda^\prime}) - \hat {\vm \beta}_{\lambda^\prime} \ts {\vf S}_\lambda \hat {\vm \beta}_{\lambda^\prime}/2 +
 \log{|{\vf S}_\lambda|_+}/2 - \text{tr}({\vf V}_{\lambda^\prime} {\vf S}_\lambda)/2 - \text{tr}({\vf V}_{\lambda^\prime} \ilpddif{l}{{\vm \beta}}{{\vm \beta}\ts})/2.
$$
The final term is then neglected, again following the PQL type assumption. 

In the case of a penalized generalized linear model, the general update (\ref{general-update}) becomes
$$
\lambda_j^* = \phi\frac{\text{tr}({{\vf S}_\lambda^-{\vf S}_j}) -
\text{tr}\{({\vf X}\ts {\vf W X} + {\vf S}_\lambda)^{-1} {\vf S}_j\}}{\hat {\vm \beta}_\lambda \ts {\vf S}_j \hat {\vm \beta}_\lambda} \lambda_j,
$$
where $\vf W$ is the diagonal matrix of weights at convergence of the usual penalized iteratively re-weighted least squares iteration used to find $\hat {\vm \beta}_\lambda$, and $\phi$ is the scale parameter, which can be estimated using the obvious equivalent of $\hat \sigma^2$. Again, under the restrictions of the Fellner-Schall method, this update corresponds to the  Schall update for the generalized linear mixed model case. 





\section{Simple examples}

\begin{figure}
\vspace*{-.6cm }

\eps{-90}{.6}{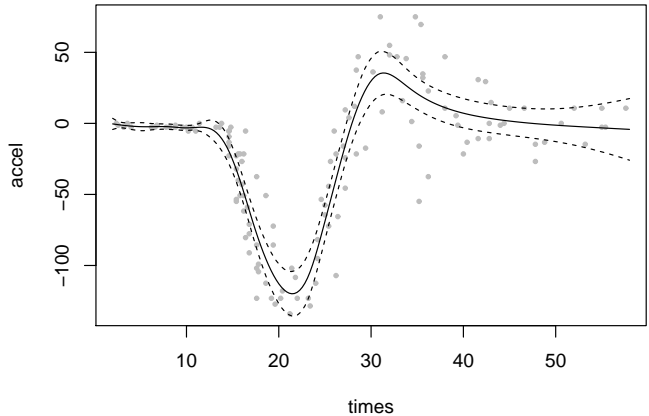}
\vspace*{-.5cm}

\caption{\small An adaptive smoother fitted to the motorcycle data using the proposed method. A fit by direct restricted marginal likelihood maximisation is indistinguishable. Previous Fellner-Schall methods could not be used for this example, as it lacks the required special structure of $S_\lambda$. \label{mcycle.fig}}
\end{figure}

This section presents two brief example applications of the generalised Fellner-Schall method developed here, which would be impossible or impractically slow with previously published versions of the method.

The first example is a simple Gaussian adaptive smooth of the motorcycle data from \cite{silverman85}, available in the MASS package \citep{MASS} in R \citep{rcore}. The data are accelerations of the head of a crash test dummy against time.  An adaptive smooth as described in \cite{wood2011} is appropriate for smoothing the acceleration data against time, with the degree of smoothness of a P-spline \citep{Eilers&Marx96} varying smoothly with time. The smooth used has five smoothing parameters with the penalties acting on overlapping subsets of the 40 model coefficients, thereby violating the structural conditions on $S_\lambda$ required by previously published Fellner-Shall iterations. The smoothing parameter optimization problem is relatively challenging as the smoothing parameters are only weakly identified from this relatively small dataset. 

The smooth was estimated using the method presented here and by the method of \cite{wood2011} using quasi-Newton optimization of the restricted marginal likelihood. In this way both methods have the same leading order computational cost per iteration, facilitating comparison. Full Newton optimization is more costly per iteration, but would require fewer iterations than quasi-Newton. Starting from all smoothing parameters set to 1, and without step length control, the new method converged in 39 steps, as against 32 for the quasi-Newton method. The fits are identical to graphical accuracy with equal effective degrees of freedom of 12.22. See Figure \ref{mcycle.fig}.

The second example is a Cox proportional hazards model for time to recurrence of colon cancer for $n=929$ patients in a chemotherapy trial \citep{moertel1995colon} available in the survival package \citep{survival.package} in R. In principle it is possible to use previously published Fellner-Schall methods for this example, by using a trick involving Poisson regression on artificially replicated data, but this entails an $O(n)$ multiplication of the computational cost, which is impractically uncompetitive with existing methods. With the update (\ref{general-update}) the cost is kept at the $O(np^2)$ that is appropriate for Cox regressions. 

The linear predictor for the Cox regression had parametric effects for whether the colon was perforated or not, obstructed or not and whether the tumour had adhered to neighbouring organs. In addition a 3 level factor indicated the control group, treatment with one drug of interest or treatment with a drug combination. Smooth effects of age were included separately for males and females along with a smooth effect for number of affected lymph nodes. For this example the new iteration, without step length control, converged in 15 steps, compared to 16 steps for direct quasi-Newton optimization using the methods of \cite{wood2015plig}. The parametric model coefficients differ only in the 4th significant digit, while differences in the estimated smooth effects are also small, as shown in Figure \ref{colon.fig}. 

   \begin{figure}
\vspace*{-.6cm }

\eps{-90}{.8}{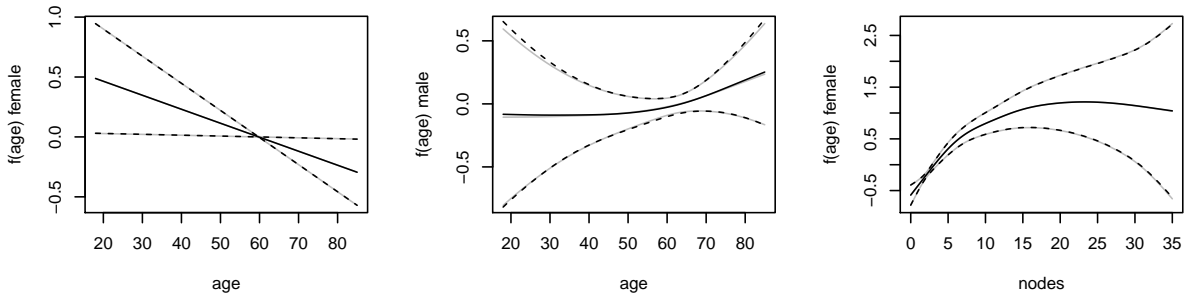}
\vspace*{-.5cm}

\caption{\small Estimated smooth effects for the colon cancer survival model, with 95\% confidence intervals. The estimates using full Laplace approximate restricted marginal partial likelihood are shown in grey, with the new method estimates overlaid in black. \label{colon.fig}}
\end{figure}

\section{A Tweedie location, scale and shape model for Mackerel \label{mack.sec}}

We now return to the motivating example, from the introduction, of modelling mackerel ({\em Scomber scombrus}) egg densities from survey data collected off the west coast of Europe in 2010. The data consist of counts of eggs in samples taken from the water column at the sampling stations shown in Figure \ref{mack-space.fig}. Available covariates are temperature and salinity at 20 m depth, water volume sampled (an offset), spatial location as longitude and latitude (converted to km east and km north), the identity of the ship collecting the data, and the sea bed depth. The latter is important as Mackerel prefer to spawn near the continental shelf edge, which occurs at a depth contour of about 200m. 

A common theme with data of this type is that the counts are highly over-dispersed relative to a Poisson distribution, but with a mean variance relationship that is less extreme than that suggested by a negative binomial distribution \cite[see e.g.][\S 5.4.1]{wood2006igam}. A \cite{tweedie1984} distribution often offers a much better characterisation of the distribution, but it would often be useful to allow the shape and scale parameters of the Tweedie distribution to vary with covariates, rather than only allowing covariates for the mean. Specifically, the Tweedie distribution assumes that the variance of random variable $y_i$ is related to its mean, $\mu_i$ via $\text{var}(y_i) = \phi_i \mu_i^{p_i}$. where the parameters $\phi_i$ and $p_i$ are parameters usually taking one fixed value for all $i$. For the mackerel data it would be useful to allow $p_i$ and $\phi_i$ to be smooth functions of covariates - particularly sea bed depth.

In particular we would like to estimate the model 
\begin{multline}
\log(\mu_i) = g_1({\tt lo}_i, {\tt la}_i) + g_2({\tt T20}_i) + g_3({\tt S20}_i) + g_4({\tt b.depth}^{1/2}) + b_{s(i)} + \log({\tt vol}_i),\\ h(p_i) = g_5({\tt b.depth}^{1/2}),~~~ \log(\phi_i) = g_6({\tt b.depth}^{1/2}), ~~~ {\tt count}_i \sim \text{Tweedie}(\mu_i,p_i,\phi_i) \label{twlss}
\end{multline}
where the $g_k$ are smooth functions, $h$ is a known link function designed to keep $1<p<2$, $s(i)$ indicates which ship collected sample $i$ and $b_{s(i)}$ are independent $N(0,\sigma^2_b)$ random effects. We represented the spatial effect using a rank 150 Duchon spline with first order derivative penalisation \citep[see][]{duchon77, miller2014}, and the other terms with rank 10 cubic penalised regression splines. The model can be estimated, {\em given smoothing parameters}, using the Newton iteration detailed in \cite{wood2015plig} and available in R package {\tt mgcv}. However the estimation of smoothing parameters using \cite{wood2015plig} would require third and fourth derivatives of the Tweedie density and these are not readily available, for the reasons given in the introduction. We therefore estimated the smoothing parameters using the iterative update (\ref{general-update}).  

Estimation converged in 13 iterations taking 17 seconds (single core of a mid range laptop computer). In comparison it took 11 seconds to fit the same model, but with fixed $p$ and $\phi$, using the method of \cite{wood2015plig} in R package {\tt mgcv}. The AIC for model (\ref{twlss}) was 180 lower than for the fixed $p$ and $\phi$ version, although residual plots (not shown) are reasonable for both models. The estimated spatial smoother is shown in Figure \ref{mack-space.fig}b, while the remaining effects are plotted in Figure \ref{mack-eff.fig}. Notice how the smooth effects of sea depth all have a pronounced peak at around $\sqrt{200}$, corresponding to the edge of the continental shelf. Both egg density and its variability appear to be peaking near the shelf edge.

\begin{figure}
\vspace*{-.6cm }

\eps{-90}{.6}{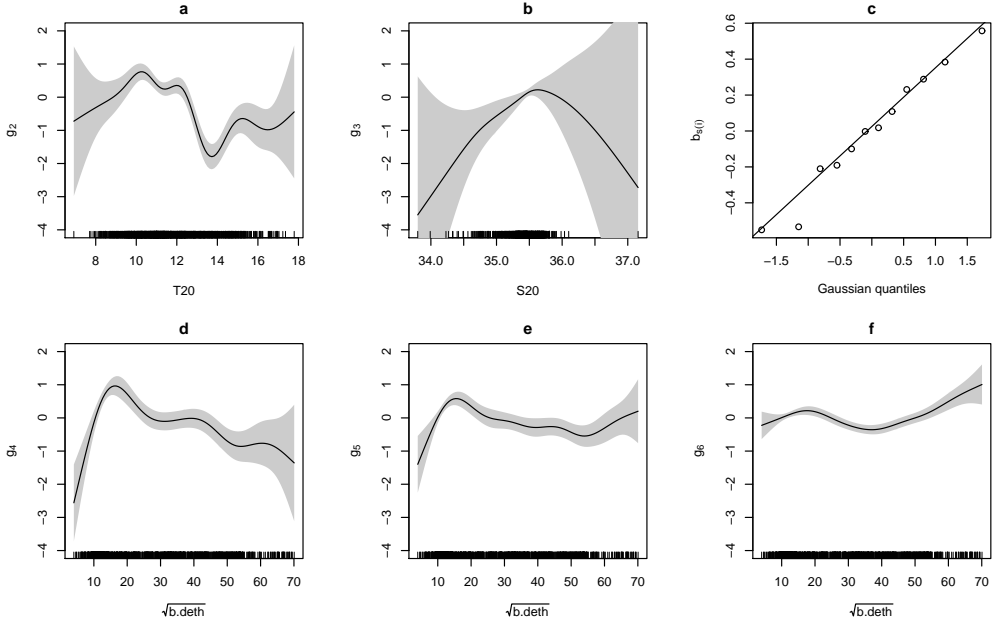}
\vspace*{-.5cm}

\caption{\small Estimated smooth effects for the Tweedie location scale and shape model of the Mackerel egg survey data discussed in section \ref{mack.sec}. Panel c shows a QQ-plot for the predicted ship level random effects. Panels d, e and f are the smooth effects of sea depth for $\mu$, $p$ and $\phi$ respectively. Notice how they all have a peak close to $\sqrt{200}$, the depth representing the continental shelf edge.  The shaded regions are approximate 95\% confidence intervals. \label{mack-eff.fig}}

\end{figure}

\section{Discussion}

Prior to the work reported here, the Fellner-Schall method could only be applied to a subset of the smooth additive models that could be estimated by direct Laplace approximate marginal likelihood maximisation. The generalizations introduced here remove this obstacle, and we have also strengthened the theoretical underpinnings of the method. The major advantage of the method is its simplicity: the direct method of \cite{wood2015plig} requires evaluation of third or fourth order derivatives of the log likelihood, which are not required by the generalized Fellner-Schall method. In addition direct optimization of the Laplace approximate marginal likelihood requires nested optimization and implicit differentiation to obtain derivatives of $\beta$ with respect to $\lambda$. Such an approach involves considerable effort if it is to be numerically stable, which is not required by the modified Fellner Schall iteration. The main theoretical cost is that, beyond the Gaussian case, we are forced to make the same simplification that underpins the PQL and performance oriented iteration methods, and neglect the dependence of the Hessian of the log likelihood on the smoothing parameters. 

As we demonstrated in section \ref{mack.sec}, our generalized Fellner-Schall method can be applied to cases in which alternative estimation methods would be very difficult to implement, but it also offers advantages in settings which are in principle less numerically taxing. The method can be applied to non-standard smooth models provided that we can obtain the first and second derivatives of the log-likelihood, which are anyway required for Newton optimization of model coefficients. This greatly simplifies the process of implementing non-standard models for particular applied problems, freeing the modeller from the more onerous aspects of implementation, to concentrate on development of the model itself. To gain insight into the effort saved, the reader might care to compare the expressions for the $4^{\rm th}$ order and second order derivatives of the generalized extreme value distribution, for example. 

Finally, an interesting question raised by the work here, is whether it is possible to reduce the implementation cost even further by replacing the Hessian of the log-likelihood in the update by a Quasi-Newton approximation, thereby allowing coefficients to be estimated by Quasi-Newton methods, and only requiring first derivatives of the log-likelihood. 

\subsection*{Acknowledgments} 

We thank Yousra El Bachir for useful comments on an earlier version of this paper. This work was funded by EPSRC grant EP/K005251/1 `Sparse, rank-reduced and general smooth modelling'. The mackerel data are available from ICES Atlantic Anguilla surveys,\\ \verb+http://eggsandlarva.ices.dk+.

\bibliography{/home/sw283/bibliography/journal,/home/sw283/bibliography/simon}
\bibliographystyle{chicago}

\end{document}